\documentclass[aps,pra,letterpaper,superscriptaddress,twocolumn]{revtex4}
\usepackage{soul}
\usepackage[T1]{fontenc}
\usepackage{amstext,amsmath,amssymb,ulem,amsthm}
\usepackage{color}

\usepackage{float} 

\usepackage{times}
\usepackage{algorithm}
\floatname{algorithm}{Protocol}
\usepackage{ftnxtra}

\usepackage{enumitem}

\usepackage{graphicx}
\usepackage[colorlinks]{hyperref}

\allowdisplaybreaks[2]

\newcounter{thm}

\theoremstyle{plain}

\newtheorem{theorem}[thm]{Theorem}

\newtheorem{lemma}[thm]{Lemma}
\theoremstyle{definition}
\newtheorem{definition}{Definition}

\newcommand{\beq}{\begin{equation}}
\newcommand{\eeq}{\end{equation}}
\newcommand{\ket} [1] {\vert #1 \rangle}
\newcommand{\bra} [1] {\langle #1 \vert}

\newcommand{\ba}{\begin{align}}
\newcommand{\ea}{\end{align}}
\newcommand{\bea}{\begin{eqnarray}}
\newcommand{\eea}{\end{eqnarray}}

\newcommand{\D}{\mathcal{D}}

\newcommand{\Z}{\mathbb{Z}}

\newcommand{\abs}[1]{\left\lvert{#1}\right\rvert}

\newcommand{\avg}[1]{\left\langle {#1} \right\rangle}

\newcommand{\DA}{\Delta_{\mathcal{A}}}

\makeatletter

\@ifundefined{textcolor}{}
{
 \definecolor{BLACK}{gray}{0}
 \definecolor{WHITE}{gray}{1}
 \definecolor{RED}{rgb}{1,0,0}
 \definecolor{GREEN}{rgb}{0,.6,0}
 \definecolor{BLUE}{rgb}{0,0,1}
 \definecolor{CYAN}{cmyk}{1,0,0,0}
 \definecolor{MAGENTA}{cmyk}{0,1,0,0}
 \definecolor{YELLOW}{cmyk}{0,0,1,0}
 }

\makeatother

\usepackage{bm}
\usepackage{mathtools}

\def\id{I}

\def\1{\mat{\id}}

\def\mat#1{\vec{#1}}

\renewcommand{\vec}[1]{\bm{\mathrm{#1}}}

\renewcommand{\sout}[1]{}

\begin{document} 
\title{Flow Ambiguity: A Path Towards Classically Driven Blind Quantum Computation}
\author{Atul Mantri}
\thanks{These authors contributed equally to this work.}
\affiliation{Singapore University of Technology and Design, 8 Somapah Road, Singapore 487372}
\affiliation{Centre for Quantum Technologies, National University of Singapore, Block S15, 3 Science Drive 2, Singapore 117543}
\author{Tommaso F. Demarie}
\thanks{These authors contributed equally to this work.}
\affiliation{Singapore University of Technology and Design, 8 Somapah Road, Singapore 487372}
\affiliation{Centre for Quantum Technologies, National University of Singapore, Block S15, 3 Science Drive 2, Singapore 117543}
\author{Nicolas C. Menicucci}
\email[]{ncmenicucci@gmail.com}
\affiliation{Centre for Quantum Computation and Communication Technology, School of Science, RMIT University, Melbourne, Victoria 3001, Australia}
\affiliation{School of Physics, The University of Sydney, Sydney, New South Wales 2006, Australia}
\author{Joseph F. Fitzsimons}
\email[]{joseph_fitzsimons@sutd.edu.sg}
\affiliation{Singapore University of Technology and Design, 8 Somapah Road, Singapore 487372}
\affiliation{Centre for Quantum Technologies, National University of Singapore, Block S15, 3 Science Drive 2, Singapore 117543}
\begin{abstract}
Blind quantum computation protocols allow a user to delegate a computation to a remote quantum computer in such a way that the privacy of their computation is preserved, even from the device implementing the computation. To date, such protocols are only known for settings involving at least two quantum devices: either a user with some quantum capabilities and a remote quantum server or two or more entangled but noncommunicating servers. In this work, we take the first step towards the construction of a blind quantum computing protocol with a completely classical client and single quantum server. Specifically, we show how a classical client can exploit the ambiguity in the flow of information in measurement-based quantum computing to construct a protocol for hiding critical aspects of a computation delegated to a remote quantum computer. This ambiguity arises due to the fact that, for a fixed graph, there exist multiple choices of the input and output vertex sets that result in deterministic measurement patterns consistent with the same fixed total ordering of vertices. This allows a classical user, computing only measurement angles, to drive a measurement-based computation performed on a remote device while hiding critical aspects of the computation.
\end{abstract}

\date{\today}
\maketitle

\section{Introduction}
\label{intro}

Large-scale quantum computers offer the promise of quite extreme computational advantages over conventional computing technologies for a range of problems spanning cryptanalysis~\cite{shor1999polynomial}, simulation of physical systems~\cite{lloyd1996universal}, and machine learning~\cite{lloyd2013quantum}. Recently, however, a new application has emerged for quantum computers: secure delegated computation \cite{Dunjko2014}.

Consider a user wishing to have a computation performed on a remote server. Two main security concerns arising for the user relate to the \emph{privacy} and the \emph{correctness} of the computation.  The privacy concern is that the description of their computation, both the program and any input data, remains hidden even from the server. The correctness concern is that a malicious server might tamper with their computation, sending them a misleading result: hence, ideally such behaviour would be detectable. Quantum protocols have been proposed that can mitigate both of these concerns. In the literature, protocols that allow for program and data privacy are known as \emph{blind} quantum computing protocols, while protocols that allow for correctness to be ensured with high probability are known as \emph{verifiable} quantum computing protocols~\footnote{Formal definitions of blindness and verifiability can be found in~\cite{Dunjko2014}.}.

The first blind quantum computing protocol was proposed by Childs~\cite{Childs2005b}. While functional, this scheme put a rather heavy burden on the client's side in terms of resources, with the client required to control a quantum memory and to perform SWAP gates. A subsequent protocol, from Arrighi and Salvail~\cite{Arrighi2006}, introduced mechanisms for both verification and blindness for a limited range of functions and can be seen as the start of an intimate link between blindness and verifiability. This link was further established with the discovery of the universal blind quantum computing (UBQC) protocol~\cite{Broadbent2009}, which allows a client, equipped only with the ability to produce single-qubit states, to delegate an arbitrary quantum computation to a universal quantum server while making it blind with \emph{unconditional security}.

This scheme has been modified and extended several times in the last few years, with works investigating robustness~\cite{Morimae2012, Sueki2013, Chien2013}, optimality~\cite{Mantri2013, Giovannetti2013, Perez2015}, and issues related to physical implementations~\cite{Morimae2012b,dunjko2012blind}. Importantly, the blind computation protocols have proven a powerful tool in the construction of verifiable quantum computing protocols, with a number of protocols emerging in recent years based on the UBQC protocol~\cite{fitzsimons2017unconditionally,Hajdusek2015, Gheorghiu2015} and on an alternative blind protocol from Morimae and Fujii~\cite{morimae2013blind} in which the client performs single-qubit measurements rather than state preparations~\cite{morimae2014verification,hayashi2015verifiable,hayashi2016self}. The relatively low overhead in such schemes has made it possible to implement both blind and verifiable quantum computing protocols in quantum optics~\cite{Barz2012,Barz2013,greganti2016demonstration}.

The question of verifiability, directly rather than as a consequence of blindness, has also attracted attention. This problem was first studied by Aharonov \emph{et~al.}~\cite{aharonov2010proceedings}, who considered the use of a constant-sized quantum computer to verify a larger device. Subsequent work by Broadbent~\cite{Broadbent2015} reduced the requirements on the prover to mirror those used in the UBQC protocol. An entirely distinct route to verification has also emerged, which considers a classical user but requires multiple entangled but noncommunicating servers~\cite{Reichardt2013,mckague2013interactive}. Surprisingly, perhaps, many of these schemes are also blind, though often this was not the aim of the paper. In fact, only a few examples of verifiable computing schemes exist that are not naturally blind~\cite{fitzsimons2015post,morimae2016post}, and it is tempting to conjecture a fundamental link between blindness and verifiability. 

The verification methods discussed above provide a very strong form of certification, amounting to interactive proofs for correctness which do not rely on any assumptions about the functioning of the device to be tested. From an experimental point of view, the first nonclassically simulable evolution of quantum systems will most likely be implemented by means of nonuniversal quantum simulators rather than fully universal quantum computers. Here, too, the problem arises of certifying the correct functioning of a device~\cite{Hauke2012} that cannot be efficiently simulated. However, in this regime interactive proofs have proven more difficult to construct. Nonetheless, progress has been made in developing a range of certification techniques for various physical systems. These include feasible quantum state tomography of matrix product states~\cite{Cramer2010}, certification of the experimental preparation of resources for photonic quantum technologies~\cite{Aolita2015}, certification of simulators of frustration-free Hamiltonians~\cite{Hangleiter2017}, and derivation of a statistical benchmark for boson sampling experiments~\cite{Walschaers2016}.

A common feature among all blind quantum computing protocols and interactive proofs of correctness for quantum computation is that they require that at least two parties possess quantum capabilities. Removing this requirement and allowing a purely classical user to interact with a single quantum server would greatly expand the practicality of delegated quantum computation since it would remove large-scale quantum networks as a prerequisite for verifiability. In the present work, we focus specifically on the question of blind computation with a completely classical client, but given the historic links between progress in blindness and verification, it is natural to expect that progress in either direction will likely be reflected in the other.

While it is presently unknown if such a protocol can exist, a negative result in this context is a  \emph{scheme-dependent} impossibility proof presented in Ref.~\cite{Morimae2014}. There, the author considered a scenario where a classical user and a quantum server exchange classical information in a two-step process. First, the classical client encodes their description of the computation using an affine encryption scheme and then sends all the classical encrypted data to the server. The server then performs a quantum computation using the received data and returns the classical output to the client who decrypts the result using their encryption key. For this setting, it was shown that secure blind quantum computing cannot be achieved unless $\text{BPP} = \text{BQP}$ (i.e.,~unless a classical computer can efficiently simulate a quantum computer). While this is an interesting result, it imposes strong assumptions on the operational method of blind quantum computation with a classical client and therefore does not seem to limit further studies in this direction. Additionally, Aaronson \textit{et}~al.~\cite{Aaronson2017} have recently suggested that information-theoretically blind quantum computing with a classical client is not likely to be possible because the existence of such a scheme implies unlikely containments between complexity classes. Additional implications that the development of a classical-client blind-computation protocol would have in complexity theory are discussed in Ref.~\cite{Dunjko2016}.

Here, we provide evidence in the opposite direction. We introduce a form of delegated quantum computation using measurement-based quantum computing (MBQC) as the underlying framework. This allows us to introduce a model-specific protocol that achieves a satisfying degree of security by directly exploiting the structure of MBQC. We show that the classical communication received by the party performing quantum operations is insufficient to reconstruct a description of the computation. This insufficiency remains even when the server is required only to identify the computation up to pre- and post-processing by polynomial-sized classical computation, under plausible complexity-theoretic assumptions. We call our scheme \emph{classically driven blind quantum computing (CDBQC)}.

The paper is structured as follows: In Sec.~\ref{chap2}, to help the readers unfamiliar with MBQC, we present a short introduction to this model of quantum computation. In Sec.~\ref{chap3} we describe the steps of the CDBQC protocol. In Sec.~\ref{chap4} we use mutual information to analyse the degree of blindness for a single round of the CDBQC protocol. In Sec.~\ref{chap5} we introduce the concept of flow ambiguity and we show how this is used by the client to hide information from the quantum server. Our conclusions are presented in Sec.~\ref{chap6}.

\section{Measurement Based Quantum Computation} 
\label{chap2}
In MBQC, a computation is performed by means of single-qubit projective measurements that drive the quantum information across a highly entangled resource state. The most general resources for MBQC are graph states~\cite{Raussendorf2001}.  A graph state is defined by a simple and undirected graph, i.e.,~a mathematical object $\mathcal{G} = (\mathcal{V},\mathcal{E})$ composed of a vertex set~$\mathcal{V}$ and an edge set~$\mathcal{E}$, with cardinality $\abs{\mathcal{V}}$ and $\abs{\mathcal{E}}$, respectively. The vertices of the graph represent the qubits, while their interactions are symbolized by the edges. A graph state $\ket{\mathcal{G}}$ is an $N$-qubit state, where $N = |\mathcal{V}|$. Each qubit is initialized in the state $\ket{+} = \frac{1}{\sqrt{2}} (\ket{0} + \ket{1})$ and then entangled with its neighbors by controlled-$Z$ gates, $\hat{C}_{Z\,{i,j}} = \ket{0}\bra{0}_i \otimes \hat{I}_j + \ket{1}\bra{1}_i \otimes \hat{Z}_j$, where $\hat{I}$ and $\hat{Z}$ are the single-qubit identity and Pauli-$Z$ gate respectively.  Explicitly, $\ket{\mathcal{G}} = \prod_{(i,j) \in \mathcal{E}} \hat{C}_{Z\,{i,j}} \ket{+}^{\otimes N}$. Equivalently, a graph state $\ket{\mathcal{G}}$ can be defined by the stabilizer relations $\hat{K}_v \ket{\mathcal{G}} = \ket{\mathcal{G}}$, with stabilizers~\cite{Raussendorf2001}
\begin{equation}
\hat{K}_v = \hat{X}_v \prod_{w \in \mathcal{N}(v)} \hat{Z}_w,\qquad \forall v \in \mathcal{V} ,
\end{equation}
where $\mathcal{N}(v)$ denotes the neighborhood of $v$ in $\mathcal{G}$. Without loss of generality, the vertices in $\mathcal{G}$ can be labeled $(1,\dotsc,N)$ in the order that the corresponding qubits are to be measured. We take this ordering to be implicit in the definition of the graph, for example as the order in which the vertices appear in the adjacency matrix for $\mathcal{G}$. It is also useful to define a specific type of graph state that will be used later. An $N$-qubit {\em cluster state} $\ket{\text{CS}}_{n,m}$ is the graph state corresponding to an $n \times m$ regular square-lattice graph $\mathcal{G}_{n,m}$. For such a graph, $N=nm$.

In the MBQC framework, given a resource state with graph $\mathcal{G}$, the standard procedure to perform a computation is to first identify two sets of qubits $\{I,O\}$ on $\mathcal{G}$. This procedure defines an \emph{open graph} $\mathcal{G}(I,O)$, such that $I,O \subseteq \mathcal{V}$ for a given $\mathcal{G}$. The set $I$ corresponds to the input set, while $O$ denotes the output set. In general, $0< \abs{I} \le \abs{O} \le \abs{\mathcal{V}}$. Note that the input and output sets can overlap. The complement of $I$ is written $I^c$, and similarly, the complement of $O$ is $O^c$. We also denote by $P(I^c)$ the power set of all the subsets of elements in $I^c$, and we define
\begin{align}
	\text{Odd}(K) \coloneqq \{ i : \abs{\mathcal{N}(i) \cap K } = 1 \mod 2 \}
\end{align}
as the odd neighborhood of a set of vertices $K \subseteq \mathcal{V}$. In this work, we are only interested in MBQC protocols that implement unitary embeddings. Hence for us, $\abs{I} = \abs{O} \le N$. Intuitively, the state of the qubits in the input set corresponds to the input state of a computation. Similarly, the qubits in the output set will contain the quantum information corresponding to the result of the computation once all the qubits in $O^c$ have been measured. In the process, the quantum information is transformed by the same principle that governs the generalized one-bit teleportation scheme~\cite{Gottesman1999,Zhou2000}.

For our purposes, we restrict the measurements to be projective measurements in the XY-plane of the Bloch sphere, denoted $M^{\alpha_j}_j = \{ \ket{\pm_\alpha} \bra{\pm_\alpha}_j \}$ for qubit~$j$, where $\ket{\pm_\alpha} = \frac{1}{\sqrt{2}} (\ket{0} \pm e^{i \alpha} \ket{1})$. As a convention, we use $b_j=0$ for the measured qubit collapsing to $\ket{+_\alpha}_j$ and $b_j=1$ for collapsing to $\ket{-_\alpha}_j$. The computation to be performed is specified both by the choice of open graph $\mathcal{G}(I,O)$ and by a vector $\vec \alpha$ specifying the measurement basis $\alpha_i$ for each qubit $i$. Note that these are the measurements \emph{that would be made directly on the cluster state if all the measurement outcomes were zero for non-output qubits}---i.e.,~if one were to implement the positive branch of the MBQC computation. Importantly, by convention, the positive branch corresponds to the target computation. In general, however, these bases need to be updated based on outcomes of earlier measurements in order to ensure the correct computation is performed. The description of the resource state, the order of measurements, and the dependency of the measurement bases on previous measurement outcomes are collectively known as a \textit{measurement pattern}.

Projective measurements are inherently random in quantum mechanics, and one needs a procedure to correct for this randomness. We show that this need for adaptation of future measurements based on previous outcomes is what prevents Bob from knowing the protocol perfectly. Not incidentally, it is also what circumvents the no-go result from Ref.~\cite{Morimae2014}. This is because only Alice knows how she is choosing to adapt future measurement bases dependent on previous measurement outcomes: Our observation is that different choices of adaptation strategy correspond to different computations in general.

The structure that determines how to recover deterministic evolution from a MBQC measurement pattern is called g-flow~\cite{Browne2007}, from \emph{generalized quantum-information flow}. Rigorously, given some resource state~$\ket{\mathcal{G}}$ and a measurement pattern on it, if the associated open graph $\mathcal{G}(I,O)$ satisfies certain g-flow conditions (to be described later), then the pattern is runnable, and it is also uniformly, strongly, and stepwise deterministic. This means that each branch of the pattern can be made equal to the positive branch after each measurement by application of local corrections, independently of the measurement angles. We use \emph{deterministic} without ambiguity to indicate all these attributes. Note also that satisfying the g-flow conditions is a necessary and sufficient condition for determinism.

In practice, the g-flow assigns a set of local Pauli corrections to a subset of unmeasured qubits after a measurement. See Ref.~\cite{Markham2013} for the fine details regarding the practicalities of g-flow. For simplicity, in the definition of g-flow below adapted from Ref.~\cite{Browne2007}, we assume all qubits are measured in the XY-plane of the Bloch sphere. The idea behind g-flow is to determine whether one can find a correction operator (related to a correcting set on the graph) that, in the case of a nonzero measurement outcome, can bring back the quantum state onto the projection corresponding to the zero outcome. This is done by applying stabilizer operators on the state. The g-flow conditions determine whether the geometrical structure of an open graph allows for these corrections after each measurement.
\begin{definition}[G-flow]
\label{def:flow}
For an open graph $\mathcal{G}(I,O)$, there exists a \emph{g-flow} $(g,\succ)$ if one can define a function ${g: O^c \to P(I^c)}$ and a partial order $\succ$ on $\mathcal{V}$ such that $\forall i \in O^c$, all of the following conditions hold:
\vspace{1ex}
\\
(G1) if $j \in g(i)$ and $j\neq i$, then $j \succ i$;\\
(G2) if $j \nsucc i$ and $i \neq j$, then $j \notin \text{Odd}(g(i))$; and \\
(G3) $i \notin g(i)$ and $i \in \text{Odd}(g(i))$.
\end{definition}
The successor function $g(i)$ indicates what measurements will be affected by the outcome of the measurement of qubit $i$, while the partial order $\succ$ should be thought of as the causal order of measurements. The condition (G1) says that if a vertex $j$ is in the correcting set of the vertex $i$, then $j$ should be measured after the vertex $i$. In other words, a correction should happen after the assigned measurement. Condition (G2) makes sure that if the correcting set of a vertex $i$ is connected to a vertex $j$, and $j$ is measured before the vertex $i$, then the vertex $j$ should have an even number of connections with the correcting set of vertex $i$. Then vertex $j$ receives an even number of equal Pauli corrections, which is equivalent to receiving none: hence, no correction can affect earlier corrections. Finally, condition (G3) certifies that each vertex $i$ has an odd number of connections with its correcting set, such that a correction is indeed performed on $i$~\cite{Markham2013}. In this sense, the g-flow conditions are understood in terms of geometrical conditions on the open graph. 

Guided by these conditions, for cluster states, here and in the following, we always adopt the same choice of vertex labeling on the graph as shown in Fig.~\ref{fig:TOCS}. This choice is motivated by our later goal of counting how many choices of open graphs satisfy the g-flow conditions on a given cluster state. Since a vertex labeling corresponds to a total order of measurement, it is easy to check that, in order to satisfy the g-flow conditions, for any vertex~$i$ the quantum information can only move towards the right, move towards the bottom, or stay on that vertex. Furthermore, condition~(G3) imposes that the information from a vertex cannot move simultaneously towards the right and towards the bottom. In order to further simplify the process of counting flows, we introduce an additional criterion, which is not strictly required by g-flow:
\vspace{1em}
\\
(G4) If $k \in \mathcal{N}(i) \cup \mathcal{N}(j)$, and if $k \in g(i)$, then $k \notin g(j)$.
\vspace{1em}

\noindent For $\mathcal{G}_{n,m}$, as we shall see later, it will prove easier to count flows satisfying (G1)--(G4) than those satisfying (G1)--(G3). This process of course only provides a lower bound on the number of flows rather than the exact number, but will be sufficient for our purposes.

With these four criteria in place, we can define a g-flow graph path, in this restricted version of g-flow, as an ordered set of adjacent edges of the graph, starting from an element of the input set and ending on an element of the output set, such that for each edge $ij$ of the path, we have $j \in g(i)$, with $j \succ i$. Then it follows that (G4) does not allow the g-flow graph paths to cross. To help the understanding of MBQC, one could think of a g-flow graph path as a representation of a wire in the quantum circuit picture. This intuition will be used later in this work to count how many ways one could define an open graph with g-flow for our choice of total ordering, which in turn provides a link to the idea that different open graphs with g-flow lead to different quantum computations.

\section{Classically Driven Blind Quantum Computation}
\label{chap3}
\label{CDBQCsec}
\begin{algorithm}
\caption{$\textbf{CDBQC}(\mathcal{G},\mathcal{A})$: Classically Driven Blind Quantum Computation}
 \label{prot:CDBQCsh}
~\\
\textbf{Protocol parameters:}
\begin{itemize}
\item A graph $\mathcal{G}$ with an implicit total ordering of vertices.
\item A set of angles $\mathcal{A}$ satisfying Eq.~\eqref{eq:domainangles}.
\end{itemize}
\textbf{Alice's input:}
\begin{itemize}
\item A target computation $\DA$ implemented using MBQC as
\beq
\nonumber
\DA^M = \{ \mathcal{G} , \vec \alpha, \vec f \},
\eeq 
representing a measurement pattern on $\mathcal{G}$ compatible with the total ordering of measurements implicit in $\mathcal{G}$, which describes a unitary embedding $\hat{U}_{A}$. The set $\vec \alpha$ represents a sequence of $N$ measurement angles over the graph $\mathcal{G}$, with each angle chosen from a set $\mathcal{A}$, which is also taken to be a parameter of the protocol and is known to both parties. The g-flow construction $\vec f$ fully determines the input state $\rho_{in}$, through the location of the  input and output qubit sets on the graph ($I$ and $O$, respectively) and the dependency sets $(\vec s^x, \vec s^z)$
\end{itemize}
\textbf{Steps of the protocol:}
\begin{enumerate}
\item \textbf{State preparation}
	\begin{enumerate}
\item Bob prepares the graph state $\ket{\mathcal{G}}$.
	\end{enumerate}
\item \textbf{Measurements}\\
	For $i= 1, \dotsc , N$, repeat the following:
	\begin{enumerate}
\item Alice picks a binary digit $r_i \in \Z_2$ uniformly at random. Then, using $r_i$, $\vec s^x$, $\vec s^z$, and the function in Eq.~\eqref{eq:alphap}, she computes the angle $\alpha_i'$. Alice transmits $\alpha_i'$ to Bob.
\item Bob measures the $i$th qubit in the basis $\{ \ket{\pm_{\alpha_i'}} \}$ and transmits to Alice the measurement outcome $b_i'\in \Z_2$. 
\item Alice records $b_i = b_i' \oplus r_i$ in ${\vec b}$ and then updates the dependency sets $(\vec s^x, \vec s^z)$. If $i \in O$, then she also records $b_i$ in ${\vec p}_{\mathcal{B}}^C$. 
	\end{enumerate}
\item \textbf{Post-processing of the output}
	\begin{enumerate}
	\item Alice implements the final round of corrections on the output string by calculating $\vec p = \vec p^C \oplus \vec s^Z_O$, with $\vec s^Z_O$ the set of $Z$ corrections on the output at the end of the protocol.
	\end{enumerate}
\end{enumerate}
\end{algorithm}

We start from the situation where Alice wants to obtain the result of a particular quantum computation. Having no quantum devices of her own, the quantum computation must have a classical output. We allow Alice to control a probabilistic polynomial-time universal Turing machine (i.e.,~a classical computer with access to randomness). Alice has classical communication lines to and from Bob, the server. Bob has access to a universal (and noiseless) quantum computer. Bob could help Alice, but she does not trust him. Alice wishes to ask Bob to perform a quantum computation for her in a way that Bob obtains as little information as possible about her choice of computation. Without loss of generality, we assume that the quantum systems used in the protocol are qubits (two-level quantum systems~\cite{Nielsen2000}). In general, here and in the following, we denote by
\beq
\label{compdescr}
\DA = \{ \rho_I, \hat{U}_A, \mathcal{M} \}
\eeq
the classical description of Alice's computation, where $\rho_I$ is the $n$-qubit input state of the computation, $\hat{U}_A$ is the unitary embedding that maps $\rho_I$ to the output state $\rho_O = \hat{U}_A \rho_I \hat{U}_A^\dagger$, and $\mathcal{M}$ is the final set of measurements on $\rho_O$ required to extract the classical output. Note that we are implying that the input state can be efficiently described classically.  For instance, it could be a standard choice of input such as the $n$-qubit computational basis $\rho_I = \ket{0} \bra{0} ^{\otimes n}$. We also (rather pedantically) assume that the number of computational steps is at most polynomial in the input size.
Making the process abstract, Alice's desired task becomes equal to sampling the string
\beq \label{eq:correctness}
{ \vec p } = \{  p_i\} = \pi (\DA) := \mathcal{M}(\hat{U}_A \rho_I \hat{U}_A^{ \dagger}) ,
\eeq
where $\pi$ is a map that describes the blind operation performed by the protocol, which outputs the correct probability distribution $\{ { p}_i \}$ on the joint measurement outcomes given $\DA$ as Alice's delegated target computation. An outline of the protocol is presented in Protocol~\ref{prot:CDBQCsh}.

Let us now introduce the relevant definitions for the variables used in the protocol and describe the steps thoroughly. The initial step of the protocol is for Bob to prepare the resource state $\ket{\mathcal{G}}$ that will be used to implement the MBQC. Once the graph state $\ket{\mathcal{G}}$ is prepared by Bob, the interactive part of the protocol starts with Alice communicating to Bob the angles to be measured, one by one. Because of the randomness introduced by the results of the projective measurements, there exists the possibility that these angles must be corrected based on the outcomes
\begin{equation}
{\vec b} := (b_1, ..., b_N) \in \mathbb{Z}_2^N
\end{equation}
of Alice's would-be measurements. Nonetheless, Alice can pick a canonical set of angles
\begin{align}
	\vec \alpha \coloneqq (\alpha_1, \dotsc, \alpha_N) \in \mathcal{A}^N
\end{align}
corresponding to the positive branch case where $\vec b=\vec 0$. As discussed earlier, it is possible that the angle for qubit~$j$ must be modified based on the outcomes of the preceding $j-1$ measurements, which we denote:
\begin{align}
	\vec b_{<j} \coloneqq (b_1, \dotsc, b_{j-1}).
\end{align}
We account for this adaptation in \emph{dependency sets}
\begin{align}
	\vec s^x &\coloneqq (s^x_1, \dotsc, s^x_N) \in \Z_2^N, \\
	\vec s^z &\coloneqq (s^z_1, \dotsc, s^z_N) \in \Z_2^N,
\end{align}
which depend on the~$\vec b_{<j}$ and also on the g-flow construction, here represented by a bit string
\begin{align}
	\vec f \coloneqq (f_1, \dotsc, f_M) \in \Z^M_2
\end{align}
of length $M$ called the \emph{flow control bits} (or just "flow bits"). At this point, we still have to quantify the value of $M$. Note, though, that it represents the number of bits needed to enumerate all the possible combinations of input and output that satisfy the g-flow conditions~\cite{Danos2006,Browne2007}. Hence, for a fixed total order of the measurements, it is a function of $N$. Explicitly, the $X$ and $Z$ corrections associated with the measurement angle of each qubit $j$ are determined by the dependency sets:
\begin{align}
\label{eq:fcb}
	s^x_j &: \D[\vec b_{<j}] \times \D[\vec f] \to \Z_2, \\
	s^z_j &: \D[\vec b_{<j}] \times \D[\vec f] \to \Z_2,
\end{align}
where the function $\D$ denotes the domain of the argument. Without loss of generality, we choose~$s^z_1 = s^x_1 = 0$ since there are no previous outcomes on which these could depend. For a fixed open graph $\mathcal{G}(I,O)$, the form of the dependency sets is uniquely defined by the g-flow~\cite{Dunjko2014}. Analogously, the flow bits~$\vec f$ fully specify the dependency sets (as functions of $\vec b$). As such, the quantum circuit that Alice intends to implement is specified by the information
\begin{align}
	(\vec \alpha, \vec f) \in \mathcal{A}^N \times \Z_2^M,
\end{align}
consisting of $N$ measurement angles and $M$ flow bits for a given graph with fixed total order of measurement. Consequently, once the graph $\mathcal{G}$ is known, there exists a one-to-one correspondence $\mathcal{G}(I,O)_{n,m} \leftrightarrow \vec f$, and we can accordingly denote the corresponding MBQC measurement pattern as follows:
\begin{align}
\DA^M = (\mathcal{G}_{n,m}, \vec \alpha, \vec f ) .
\end{align}
Explicitly, note that by choosing $\vec f$, Alice is defining a unique choice of the input and output on the graph state before the protocol begins. In line with the computation description from Eq.~\eqref{compdescr} we call $\rho_I$ the input state on $\mathcal{G}$.

We now turn our attention to what kind of information Bob receives when Alice asks him to perform the measurements on her behalf. The interactive part of the protocol consists of $N$ steps. At each step $i$, Alice requests Bob to measure, in the XY-plane of the Bloch sphere, the $i$th qubit, according to the total order implied by $\mathcal{G}$, and he sends back a bit for each measurement. We identify the measurement instructions Bob receives as a list of angles
\begin{align}
	\vec \alpha' \coloneqq (\alpha'_1, \dotsc, \alpha'_N) \in \mathcal{A}^N ,
\end{align}
and we label the string of bits Alice receives from Bob as
\begin{align}
	\vec b' \coloneqq (b'_1, \dotsc, b'_N) \in \Z_2^N ,
\end{align}
while remembering that they are communicated alternately ($\alpha'_1$~to Bob, $b'_1$~to Alice, $\alpha'_2$~to Bob, $b'_2$~to Alice, etc.). Note that in the case of a dishonest Bob, the string $\vec b'$ does not need to correspond to real measurement outcomes but could have been generated by Bob through some alternative process. 

Realizing that measuring $\alpha$ can just as easily be effected by asking Bob to measure $\alpha + \pi$ and then flipping the returned outcome bit, we introduce a \emph{uniformly random} $N$-bit string
\begin{align}
	\vec r := (r_1, \dotsc, r_N) \in \Z_2^N
\end{align}
that Alice will use to pad the angles in an attempt to conceal the measurement outcomes. All that remains is to specify how $\vec \alpha'$ depends on $\vec \alpha$. This is specified by the following functional dependence~\cite{fitzsimons2017unconditionally, Danos2006, Danos2007}:
\begin{align}
\label{eq:alphap}
	\vec \alpha' = (-1)^{\vec s^x} \vec \alpha + (\vec s^z + \vec r) \pi \mod 2\pi,
\end{align}
which follows from the g-flow construction and shows how corrections change subsequent measurement angles. Here we have used multi-index notation to present the result concisely as a vector. Note that the dependency sets $(\vec s^x$, $\vec s^z)$ are updated by Alice after each measurement. To make the analysis of the protocol meaningful, we construct a domain for~$\vec \alpha$ such that the domain of all valid~$\vec \alpha'$ is the same. Thus, in general
\begin{align}
\label{eq:domainangles}
	\mathcal{A} = \bigl\{(-1)^x \theta + z \pi \, \big|\, \theta \in \mathcal{A}, x \in \Z_2, z \in \Z_2 \bigr\}.
\end{align}
Also note that now
\begin{align}
	\vec b' = \vec b \oplus \vec r,
\end{align}
where $\oplus$ indicates addition modulo $2$ for each bit. We can identify the \emph{data} that Bob receives during the interactive part of the protocol (some from Alice, some from his own measurements) as:
\begin{align}
	\text{(data Bob receives)} \coloneqq (\vec b', \vec \alpha') \in \Z_2^N \times \mathcal{A}^N.
\end{align}

The interactive part of the protocol ends when all the qubits have been measured and Alice holds the binary register $\vec b$, derived from $\vec b'$ to account for the one-time pad $\vec r$. Since Alice knows the output set $O$, whenever the $i$th qubit belongs to the set of output qubits, Alice saves $b_i$ into a second binary sequence of length $|O|$:
\begin{align}
 {\vec p}_{\mathcal{B}}^C \coloneqq (p_1, \dotsc, p_{|O|}) \in \Z^{|O|}_2 , 
\end{align}
where $p_i = b_i$, $\forall i \in O$. If Bob is honest, then ${\vec p}_{\mathcal{B}}^C$ is equivalent to ${\vec p}^C$. At the end of the protocol, this string contains the classical result of the computation, up to classical post-processing. This is accounted for by calculating ${\vec p} = {\vec p}^C \oplus {\vec s}^Z_O$, where ${\vec s}^Z_O$ is used to represent the final set of $Z$ corrections on the output qubits. Clearly, the classical nature of the client allows us to consider only quantum computations with classical output.

In order for the protocol to have any utility, we require that the output ${\vec p}$ satisfies Eq.~\eqref{eq:correctness}, a property known as correctness. The correctness of this protocol can be proved straightforwardly. Note that the positive branch of the MBQC pattern $\DA^M$ [that is, where all the measurement outcomes happen to be equal to zero ($\vec b = \vec 0$)] implements Alice's target computation $\DA$ by definition. In the circuit model, this corresponds to a quantum circuit that implements the unitary $\hat{U}_{A}$ over the correct input state $\rho_I$ and a final round of measurements whose output is the binary string ${\vec p}$~\cite{Raussendorf2003}. Below, we give a proof of the correctness of the CDBQC protocol.
\begin{theorem}[Correctness]
\label{theo:correctsh}
For honest Alice and Bob, the outcome of Protocol~\ref{prot:CDBQCsh} is correct.
\end{theorem}
\begin{proof}
There are only two differences between Protocol~\ref{prot:CDBQCsh} and a conventional MBQC implementation of $\DA^M$. The first is the use of $\vec{r}$ to hide measurement outcomes. The effect of $\vec{r}$ is to add an additional $\pi$ to the measurement angle on certain qubits, resulting in a bit flip on the corresponding measurement result $b_i'$. However, since this is immediately undone, it has no effect on the statistics of the measurement results obtained after decoding $\vec{b}$.

The other difference is that the g-flow construction, and hence the dependency sets, is only known to Alice and not to Bob. However, this does not affect the input state, which is equivalent to the usual case if Alice is honest (i.e., if she correctly performs her role in implementing the protocol). Furthermore, if Alice updates the measurement angles correctly using the dependency sets as dictated by the g-flow, and Bob measures them accordingly, every branch of $\DA^M$ is equivalent to the positive branch. Then the measurement pattern correctly implements the unitary transformation $\rho_{out} = \hat{U}_{A} \rho_{in} \hat{U}_{A}^\dagger$. The protocol also allows Alice to identify the elements of the output string $ \vec{p}^C $ in $ \vec b $, since she knows the position of the output on the graph. Hence, when both Alice and Bob follow the protocol, the output string $ \vec{p} = {\vec p}^C \oplus {\vec s}^Z_O$ is the desired probability distribution that follows from the joint measurement of the correct quantum output.
\end{proof}

\section{Blindness analysis}
\label{chap4}
We now look at the degree of blindness for a single round of Protocol~\ref{prot:CDBQCsh}. In this setting, we consider a cheating Bob with unbounded computational power, able to deviate from the protocol and follow any strategy allowed by the laws of physics. Our aim, however, is not to verify that Bob is indeed performing the correct quantum computation as requested. Instead, we want to quantify the amount of information that Bob can access when Protocol~\ref{prot:CDBQCsh} is run only once (stand-alone) and compare it against the total amount of information needed to describe the computation. To completely identify Alice's computation, Bob needs to know the description $\DA^M$.

We identify variables with uppercase letters and particular instances of such variables with lowercase letters. The probability of a given instance~$\vec x$ of a random variable~$\vec X$ is denoted~$\Pr(\vec x)$, and averaging over~$\vec X$ is denoted~$\avg \cdot_{\vec X}$ or~$\avg \cdot$ when there is no ambiguity. Given a random variable $\vec X$, we call $N_{\vec X}$ the number of possible outcomes for the variable and $n_{\vec X} \coloneqq \log_2 N_{\vec X}$ the number of bits required to enumerate them. 

We denote the Shannon entropy~\cite{cover2012elements} of a random variable~$\vec X$ by $H(\vec X) \coloneqq \avg{-\log_2 \Pr(\vec x)}_{\vec X} \leq n_{\vec X}$, with equality if and only if $\vec X$ is uniformly random. For two random variables~$\vec X$ and~$\vec Y$, their joint entropy is written $H(\vec X, \vec Y) \coloneqq \avg{-\log_2 \Pr(\vec x,\vec y)}_{\vec X, \vec Y}$, and the conditional entropy of $\vec X$ given~$\vec Y$ is $H(\vec X|\vec Y) \coloneqq \avg{-\log_2 \Pr(\vec x|\vec y)}_{\vec X, \vec Y}$. These satisfy
\begin{align}
\label{eq:condent}
	H(\vec X|\vec Y) = H(\vec X, \vec Y) - H(\vec Y).
\end{align}
The mutual information of~$\vec X$ and~$\vec Y$ is
\begin{align}
	I(\vec X ; \vec Y)
	&\coloneqq H({\vec X}) + H({\vec Y}) - H(\vec X, \vec Y)
\nonumber \\
	&= H({\vec X}) - H(\vec X| \vec Y)
\nonumber \\
	&= H({\vec Y}) - H(\vec Y| \vec X),
\end{align}
which will be our main tool of analysis. Intuitively, ${I(\vec X ; \vec Y)}$ measures how much information~$\vec Y$ has about~$\vec X$. More precisely, it quantifies how much the entropy of~$\vec X$ is reduced, on average, when the value of~$\vec Y$ is known. Because of the symmetry of the definition, these statements also hold when the roles of $\vec X$ and $\vec Y$ are swapped.

Let us call the angles variable~$\vec A$ and the flow variable~$\vec F$. Specifying $\DA^M$, in general, therefore requires $n_{\vec A} + n_{\vec F}$ bits. In addition, we use~$\vec B$ for the eventual measurement outcomes and $\vec R$ for the (uniformly random) string of $\pi$-shift bits that is known only to Alice. In any given run of the protocol, $\vec A$ and $\vec F$ are drawn from a joint prior probability~$\Pr(\vec \alpha, \vec f)$, which is known to Bob. Thus $H(\vec A, \vec F) \leq n_{\vec A} + n_{\vec F}$ bits, with equality if and only if the prior is uniform over~$\vec F$ and~$\vec A$. Note that we do not make any assumptions about this prior in what follows.

We have seen before that a single instance of the data Bob receives at the end of Protocol~\ref{prot:CDBQCsh} is equal to $(\mathcal{G}, \vec b', \vec \alpha')$. In a stand-alone setting, this is the only data available to Bob from which he might be able to gain some information about the circuit chosen by Alice. If this protocol were to be used as a subroutine or in parallel with another protocol, then one must analyze the security in a composable framework. Such an analysis is beyond the scope of the present work and is left as an open problem. Note that the graph is considered a parameter of the protocol and not part of Alice's secret. Bob's useful information at the end of a single run of Protocol~\ref{prot:CDBQCsh} is then equal to the mutual information $I(\vec B', \vec A' ; \vec A, \vec F)$ between the variables associated with the circuit $(\vec A, \vec F)$ and Bob's data $(\vec B', \vec A')$.

In other words, we are modeling the leakage of information as an unintentional classical channel between Alice and Bob, where $(\vec A, \vec F)$ is the input of the channel and $(\vec B', \vec A')$ is the output at Bob's side. Then, the mutual information tells us how many bits of the original message Bob receives on average, when averaged over many uses of the channel. Importantly, one cannot recover from mutual information what bits of the original message are passed to Bob. For our protocol, the mutual information satisfies the following bound, which does not rely on any computational assumption but it is entirely derived from information theory.

\begin{theorem}[Blindness]
\label{theo:blindnesssh}
In a single instance of Protocol~\ref{prot:CDBQCsh} the mutual information between the client's secret input $\{\vec \alpha, \vec f\}$ and the information received by the server is bounded by
\begin{equation}
I(\vec B', \vec A' ; \vec A, \vec F) \leq H({\vec A'}).\label{eq:thm2}
\end{equation}
\end{theorem}
\begin{proof}
From the definition of mutual information, we have 
\begin{align}
&I(\vec B', \vec A' ; \vec A, \vec F) \nonumber = H(\vec B', \vec A') - H(\vec B', \vec A'|\vec A, \vec F).
\end{align}
Applying the inequality $H(\vec X,\vec Y) \leq H(\vec X) + H(\vec Y)$, together with the fact that $H(\vec B') \leq n_{\vec B'} = N$, to the above equation yields
\begin{align}
I(\vec B', \vec A' ; \vec A, \vec F) &\leq H(\vec A') + N - H(\vec B', \vec A'|\vec A, \vec F).
\end{align}
What remains to be shown is that $H(\vec B', \vec A'|\vec A, \vec F) \geq N$. This result is proved as Lemma~\ref{lem:condent} in the Appendix~\ref{app:jointprob} by bounding $\text{Pr}(\vec b', \vec \alpha'|\vec \alpha, \vec f) \leq 2^{-N}$ based on the full joint probability for the protocol. With this bound in place, Eq.~\eqref{eq:thm2} directly follows.
\end{proof}

The conditional entropy $H(\vec A, \vec F|\vec B', \vec A')$ quantifies the amount of information that, on average, remains unknown to Bob about Alice's computation at the end of Protocol~\ref{prot:CDBQCsh}. As mentioned previously, in the case where Alice chooses the measurement angles~$\vec A$ uniformly randomly from a finite set, one~$A_j$ for each qubit, and she chooses the flow $\vec F$ uniformly randomly from the set of all flows compatible with the total order implicit in $\mathcal{G}$, then
\begin{equation}
H(\vec A, \vec F) = n_{\vec A} + n_{\vec F}. \label{eq:uni_a_f}
\end{equation}
In this case by calculating the conditional entropy 
\begin{align}
H(\vec A, \vec F | \vec B', \vec A') 
	&= H(\vec A, \vec F) - I(\vec B', \vec A' ; \vec A, \vec F)
	,
\end{align}
we have $H(\vec A, \vec F | \vec B', \vec A')  \ge n_{\vec F}$ because of Theorem~\ref{theo:blindnesssh}. Note that Theorem~\ref{theo:blindnesssh} guarantees zero mutual information for a single run of Protocol 1 only if $n_{\vec A} = 0$, which means only one choice of measurement angle for each qubit. However, the structure of the domain of $\vec \alpha$ and $\vec \alpha'$ [see Eq.~\eqref{eq:domainangles}] forbids such a choice. A minimal choice of angles that is not classically simulable (via the Gottesman-Knill theorem~\cite{gottesman1998heisenberg}) is given by 
\begin{equation}
\label{angdom}
\mathcal{A} = \left\{ \frac{\pi}{4}, \frac{3\pi}{4},\frac{5\pi}{4} ,\frac{7\pi}{4} \right\} .
\end{equation}
In this case, for each angle $\alpha_j$, one has $n_{\alpha_j} = 2$, so $n_{\vec A} = 2N$. Since $H(\vec A') \leq n_{\vec A'} = n_{\vec A}$, Bob gains at most two bits of information per qubit measured, with this information being a nontrivial function of both $\vec \alpha$ and $\vec f$.

\section{Application to cluster states}
\label{chap5}
\begin{figure}
\begin{center}
\includegraphics[width=0.7\linewidth]{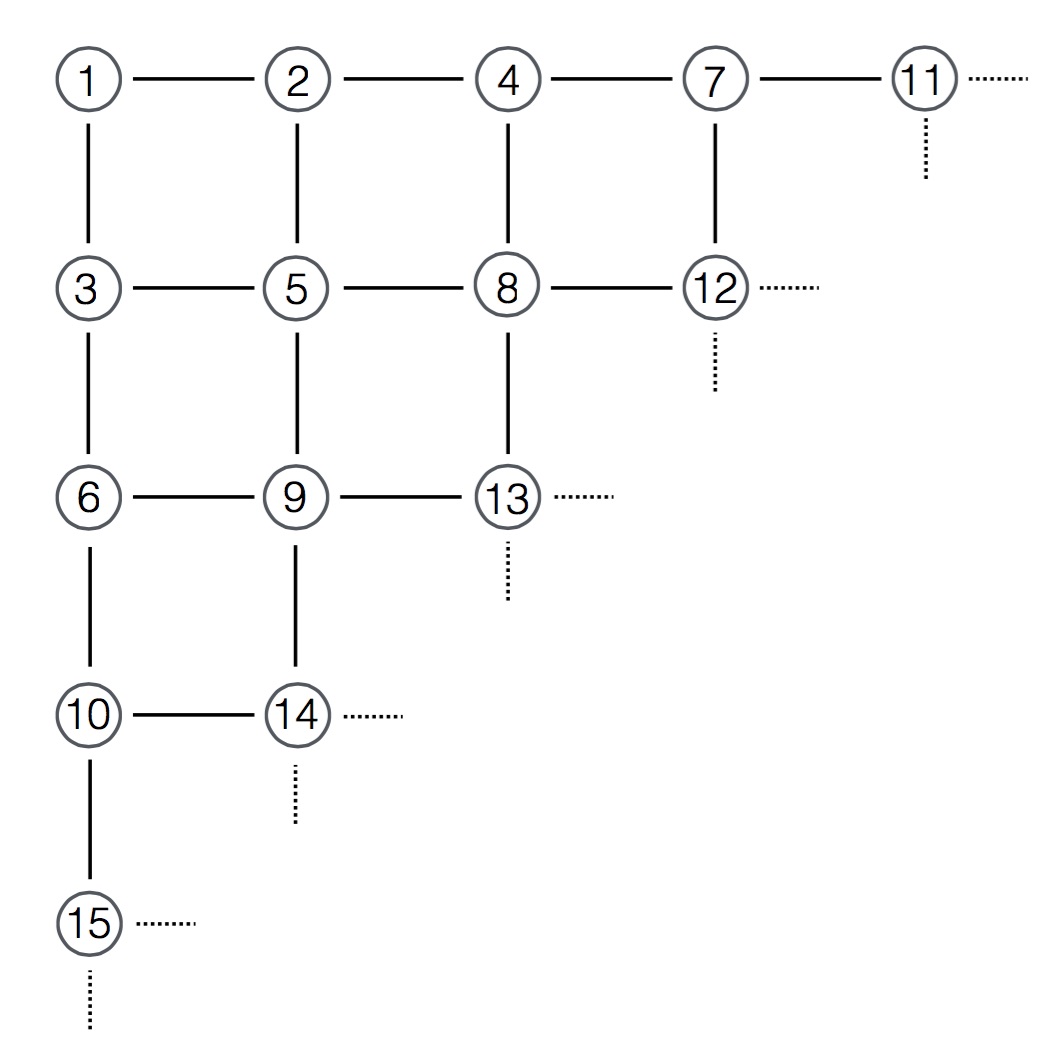}
\end{center}
\vspace*{-5mm}
\caption{Total order of the measurements for a generic $n\times m$ cluster state used as a resource state in Protocol~\ref{prot:CDBQCsh}.\label{fig:TOCS}}
\end{figure}

To conclude the security analysis of the stand-alone scenario, it is necessary to calculate the exact value of $N_{\vec F}$, which in turn gives us the value of $n_{\bf F}$ and hence the lower bound of the conditional entropy for the case of uniform variables $\vec{A}, \vec{F}$ as explained above. Clearly this depends on the choice of $\mathcal{G}$. Here, we consider the case of cluster states, where $\mathcal{G}$ is taken to be $\mathcal{G}_{n,m}$ with implicit total ordering of vertices as illustrated in Fig.~\ref{fig:TOCS}. Note that $M$, the length of the bit string $\vec f$, is equal to $\log_2 N_{\vec F}$. When condition (G4) is included, the g-flows we consider correspond to \textit{focused g-flows}~\cite{mhalla2011graph}. Hence, there is a one-to-one correspondence between an instance of a g-flow $\vec f$ of $\vec F$ and a choice of input and output set on the graph~\cite{mhalla2011graph}. Here, we place a lower bound on $M$ by counting flows that satisfy conditions (G1)--(G4). The use of the additional constraint (G4), which is not implicit in the definition of g-flow, implies that we are undercounting the total number of flows: hence,
\begin{align}
\label{eq:numflows}
	N_{\vec F} \geq \# \mathcal{G}(I,O)_{n,m},
\end{align}
where $\# \mathcal{G}(I,O)_{n,m}$ corresponds to the number of possible ways one can define an open graph that satisfies conditions (G1)--(G4). We now show that this quantity can grow exponentially in the dimensions of the cluster state such that $n_{\vec F} \propto N$.
\begin{theorem}
\label{theo:flows}
For a cluster state corresponding to $\mathcal{G}_{n,m}$ with fixed total order as depicted in Fig.~\ref{fig:TOCS}, the number of different open graphs $\mathcal{G}(I,O)$ satisfying conditions (G1)--(G4) is given by
\begin{align}
\#\mathcal{G}(I,O)_{n,m} &= F_{2\min(n,m)+1}^{|n-m|} \prod_{\mu=2}^{\min(n,m)} F_{2\mu}^2 . \label{eq:numberofflows}
\end{align}
where $F_i$ is the $i$th Fibonacci number.
\end{theorem}
\begin{proof}
The proof of this theorem is somewhat involved. We begin by considering a set of diagonal cuts across $\mathcal{G}_{n,m}$, as depicted in Fig.~\ref{fig:cuts}(a). As we are considering only those flows that satisfy condition (G4), there is a straightforward constraint on the information flow, which can be seen by isolating a single cut and the vertices linked by edges that the cut passes through [see Fig.~\ref{fig:cuts}(b)]. In the following discussion, we consider only the vertices connected by edges through which a particular cut passes. Because of the total ordering imposed on the vertices of $\mathcal{G}_{n,m}$, conditions (G1)--(G3) ensure that information can only pass through a cut from the left side to the right side and not in the reverse direction. Condition (G4) then allows exactly the set of flows where for any vertex $k$ on the right side of the cut, information flows to $k$ from at most one of its neighbors on the left side of the cut. So, if $i,j \in \mathcal{N}(k)$, then $k \notin g(i)\cap g(j)$.

\begin{figure}
\begin{center}
\includegraphics[width=0.7\columnwidth]{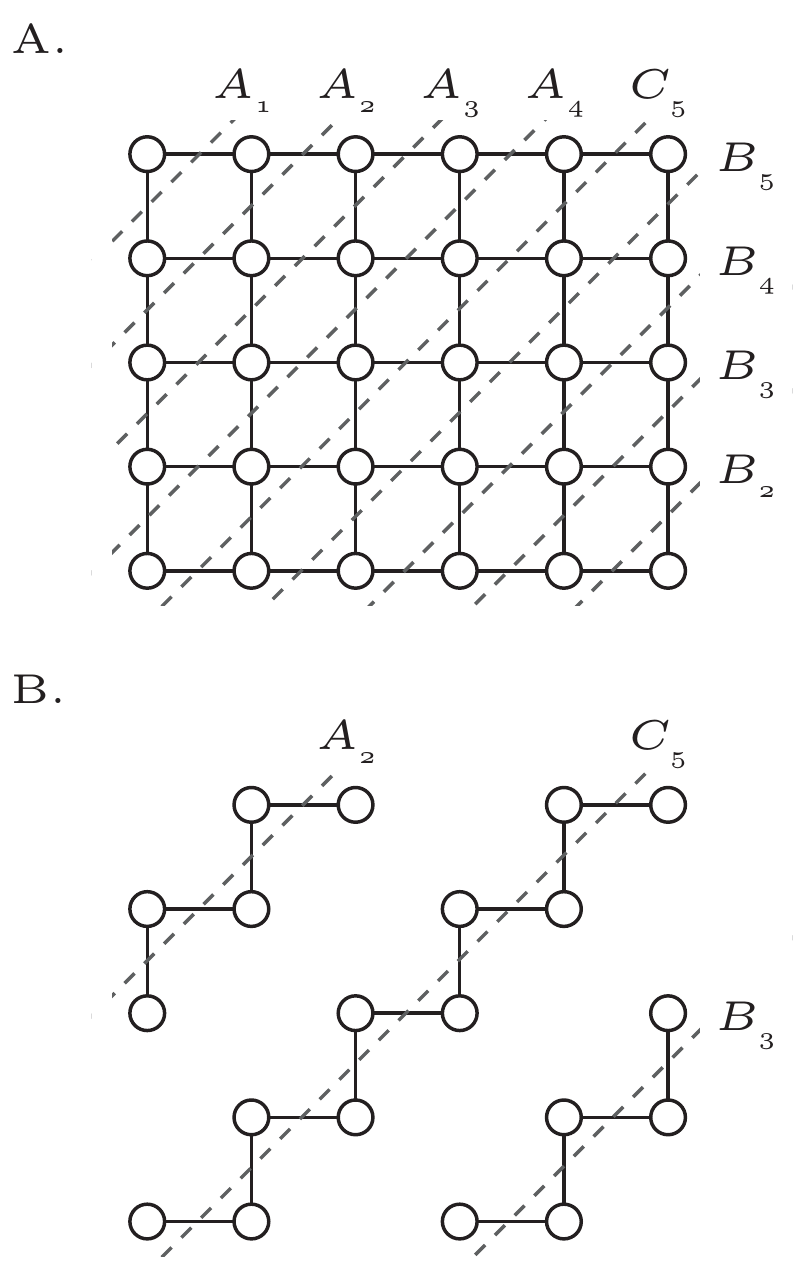}
\end{center}
\caption{(a)~A cluster-state graph $\mathcal{G}_{n,m}$ with diagonal cuts imposed. The flow across each cut is independent, and the number of possible flows across each cut is indicated. (b)~Several cuts with their neighboring vertices isolated. \label{fig:cuts}}
\end{figure}

We divide the cuts into three types: (i)~those having one less neighboring vertex on the left side of the cut than on the right, (ii)~those having one more neighboring vertex on the left side of the cut than on the right, and (iii)~those having an equal number of neighboring vertices on both sides of the cut. We label the total number of flows for each type as $A_\mu$, $B_\mu$, and $C_\mu$, respectively, where $\mu$ indicates the number of neighboring vertices on the left side of the cut, as shown in Fig~\ref{fig:cuts}(b).

In order to quantify these, we begin by noting that $A_\mu = A_\mu^{\to} + A_\mu^{\not\to}$, where $A_\mu^{\to}$ denotes the number of flows with the restriction that information flows from the uppermost neighboring vertex on the left side of the cut to the uppermost neighboring vertex on the right side of the cut, and $A_\mu^{\not\to}$ denotes the number of flows where this constraint is not satisfied. These quantities can be calculated using a simple recursion relation, as follows.

Here, $A_\mu^{\to}$ allows precisely one possibility for flow between the uppermost vertices of the cut, precluding flow from the uppermost vertex on the left side of the cut to lower vertices on the right side. Hence the remaining $\mu-1$ vertices on the left side and $\mu-1$ on the right side will be isolated and identical to the situation where the cut partitions one fewer vertex on each side. Thus, $A_\mu^{\to} = A_{\mu-1} = A_{\mu-1}^{\to} + A_{\mu-1}^{\not\to}$.

Calculating $A_\mu^{\not\to}$ is a little more involved, as there are two possibilities to consider. The first is that no information flows from the uppermost vertex on the left side of the cut across the cut (in which case it is an output). In this case, isolation of the lower vertices occurs as in the analysis of  $A_\mu^{\to}$; hence, there are $A_{\mu-1}^{\to} + A_{\mu-1}^{\not\to}$ possible flows. In the second case, no information can flow into the second uppermost vertex on the right side of the cut; hence, only $A_{\mu-1}^{\not\to}$ flows are possible. Thus, $A_{\mu}^{\not\to}  = A_{\mu-1}^{\to} + 2A_{\mu-1}^{\not\to} = A_{\mu}^{\to}+A_{\mu-1}^{\not\to}$. These correspondences are depicted in Fig.~\ref{fig:induction}.

Note that $A_{\mu}^{\to}$ and $A_{\mu}^{\not\to}$ satisfy the same recursion relation as the Fibonacci sequence when ordered as $(A_{1}^{\to}, A_{1}^{\not\to}, A_{2}^{\to}, A_{2}^{\not\to}, \dotsc)$ and starting with $A_{1}^{\to} = 1 = F_2$ and $A_{1}^{\not\to} = 2 = F_3$. It follows that $A_\mu = F_{2\mu+2}$. By similar arguments, we have $B_\mu = F_{2\mu}$ and $C_\mu = F_{2\mu+1}$.

\begin{figure}
\begin{center}
\includegraphics[width=\columnwidth]{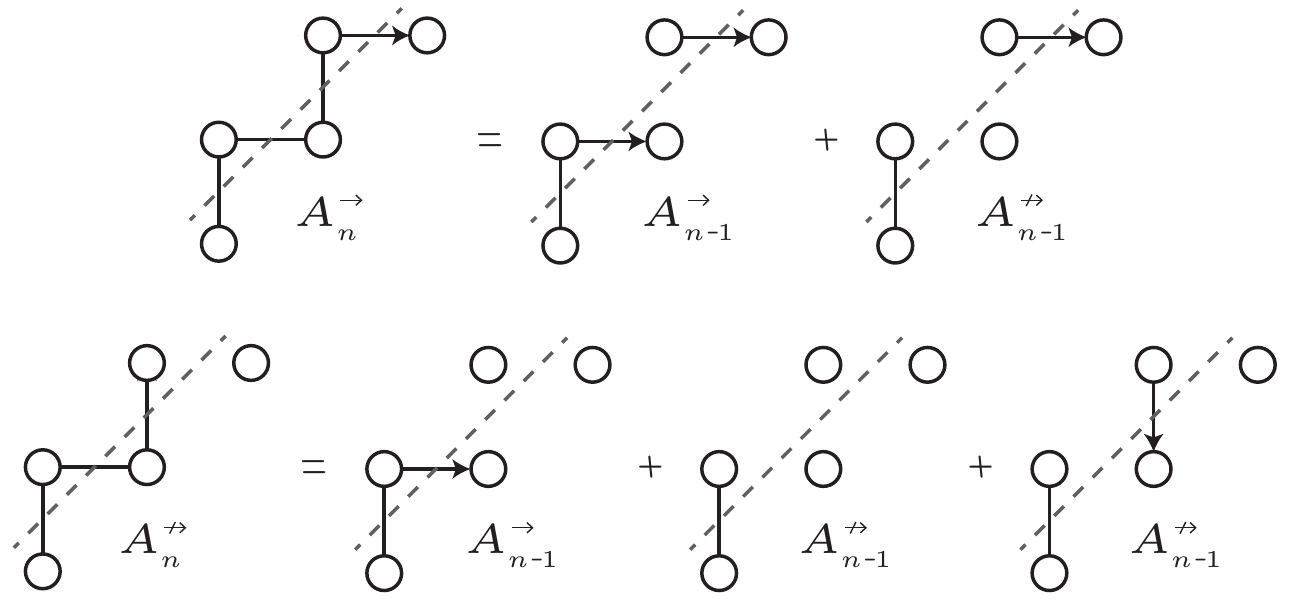}
\end{center}
\caption{$A_{\mu}^{\to}$ and $A_{\mu}^{\not\to}$ can be constructed recursively, as shown above. Here arrows indicate information flow, while edges indicate the possibility of information flow.\label{fig:induction}}
\end{figure}

It remains only to be noted that the configuration of the information flow across one cut is independent of the information flow across other cuts; hence, the total number of possible flows is given by the product of the possible flows across each cut. Therefore,
\begin{align}
&\#\mathcal{G}(I,O)_{n,m} =\nonumber \\
& \Bigg(\prod_{\mu=1}^{\min(n,m)-1} F_{2\mu + 2}\Bigg)  F_{2\min(n,m)+1}^{|n-m|} \Bigg(\prod_{\nu=2}^{\min(n,m)} F_{2\nu}\Bigg) ,\label{eq:gio-exact}
\end{align}
which simplifies to Eq.~\eqref{eq:numberofflows} as required.
\end{proof}

The Fibonacci numbers can be written exactly in terms of the golden ratio $\phi = {\frac{1}{2} (1+\sqrt 5)}$ as
\begin{equation}
F_k = \frac{\phi^k - (-\phi)^{-k}}{\sqrt{5}}.
\end{equation}
For large cluster states ($n, m \gg 1$), the number of possible flows is given by
\begin{align}
\#\mathcal{G}(I,O)_{n,m} &\approx 5^{-\frac{|n-m|}{2}}  \phi^{(2\lambda+1)|n-m|} \prod_{\nu=2}^{\lambda} \frac{\phi^{4\nu}}{5}\\
&= 5^{-\frac{(n+m-2)}{2}}  \phi^{2mn+ m + n - 4},
\end{align}
where $\lambda = \min(n,m)$. The above approximation is obtained by noting that $F_k \approx \frac{\phi^k}{\sqrt{5}}$, since $\abs{(-\phi)^{-k}} \ll 1$ for large $k$, and using this to approximate Eq.~\eqref{eq:numberofflows}.

Taking $N=nm$ and assuming $m$ grows polynomially in $n$, then $m=\text{poly}(n)$, and
\begin{align}
	\#\mathcal{G}(I,O)_{n,m} = 2^{2 N \log_2 \phi + O(N^\epsilon)}
\end{align}
for some $\epsilon<1$. In such a case, using Eq.~\eqref{eq:numflows} and evaluating to leading order,
\begin{align}
\label{eq:nF}
	 n_{\vec F} \geq \log_2 \#\mathcal{G}(I,O)_{n,m} \approx 1.388\,N.
\end{align}
This result implies that the conditional entropy $H(\vec A, \vec F | \vec B', \vec A') \geq 1.388\,N$. For the case of a computation chosen uniformly at random by Alice, the total number of bits required to entirely describe her computation is approximately equal to $3.388\, N$. However Bob only receives exactly $2N$ bits of information from Alice (the angles $\vec \alpha'$). From Theorem 2.1 in Ref.~\cite{Ambainis1999}, it is easy to verify that Bob cannot decode Alice's computation entirely with unit probability. Additionally, Theorem 2.4 in Ref.~\cite{Nayak1999} shows that Bob cannot guess Alice's computation with probability greater than $2^{-1.388 N}$.

To make sense of this result, one should remember that a particular \emph{deterministic} MBQC computation is characterized by identifying an input and output set on the underlying graph of the resource state, together with an information flow construction. This structure determines how the quantum information is deterministically transferred via projective measurements from the physical location of the input to the output. Furthermore, once the input and output systems are fixed on the graph, the flow, if it can be constructed, is unique. Hence, in the canonical approach to MBQC, the usual procedure is to fix the input and the output and assign a partial order of measurements that guarantees determinism under a specific set of rules. Consequently, the flow construction imposes a total order of measurements, which must respect the partial one.

\begin{figure}
\begin{center}
\includegraphics[width=\linewidth]{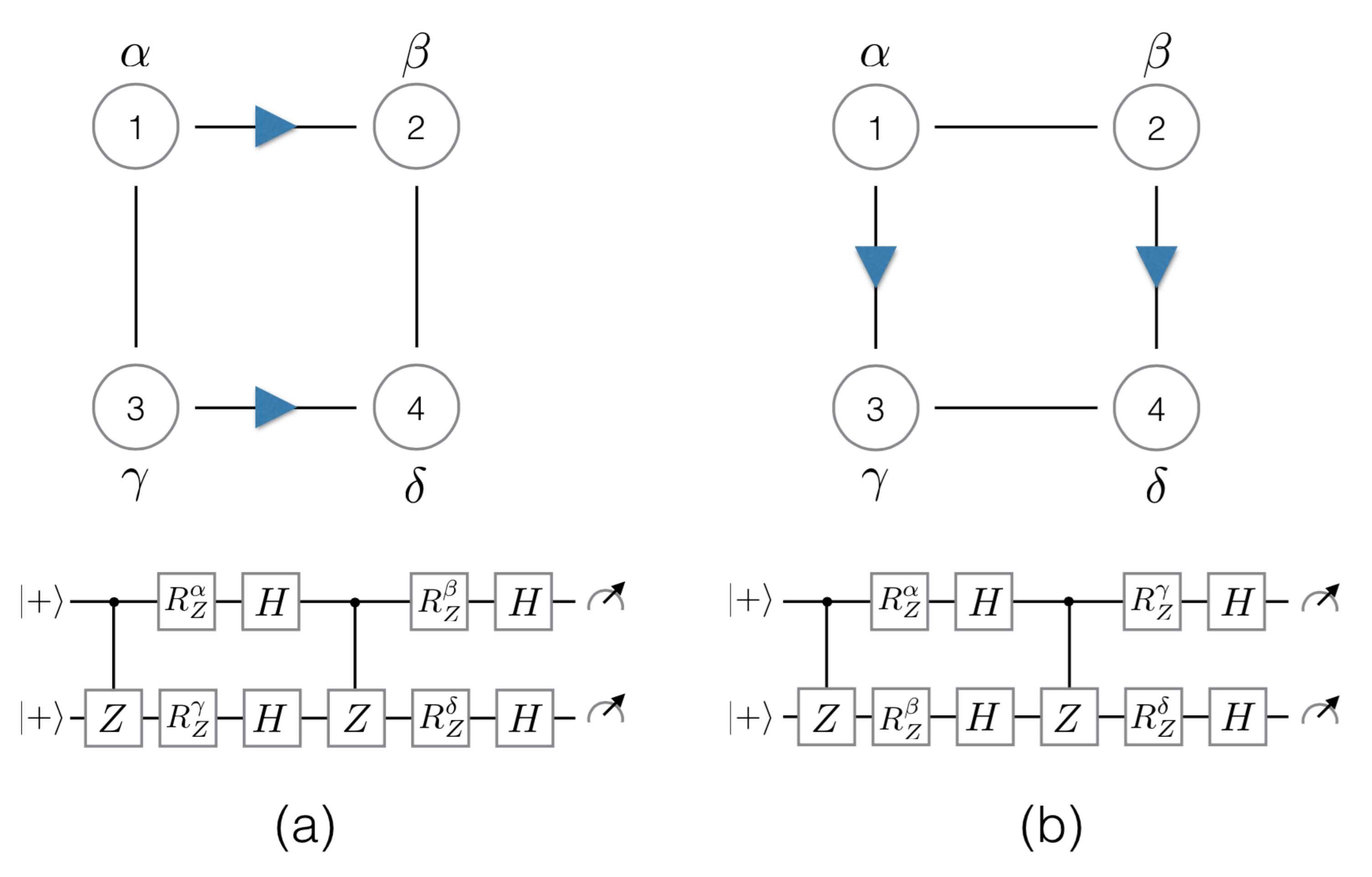}
\caption{A $2 \times 2$ cluster state with measurement angles $\{ \alpha, \beta, \gamma, \delta \}$. In this example we show how to encode two different computations using a fixed total order of measurements $\{1,2,3,4\}$. The difference follows from the choice of the $\mathcal{G}(I,O)$. In diagram~$(a)$, the input set is $\{1,3 \}$, and the output set is $\{2,4\}$, with g-flow function $g(1)=\{2\}$, and $g(3)=\{4\}$. The equivalent circuit associated with the positive branch of this MBQC pattern is shown below. Note that any final round of corrections is pushed into the classical post-processing of the output. In diagram~$(b)$, the input set is $\{1,2 \}$, the output set is $\{3,4\}$, $g(1)=\{3\}$, and $g(2)=\{4\}$. Similarly to (a), we show the circuit of the positive branch of the MBQC pattern, and the final round of corrections is classically post-processed.\label{fig:5}
}
\end{center}
\end{figure} 

Here, we \emph{have reversed} this point of view. As such, Theorem~\ref{theo:blindnesssh} is based on the nontrivial observation that, for a given MBQC resource state with a fixed total order of measurements, choices of g-flow, i.e. choices of input and output vertices on the graph that correspond to different deterministic quantum computations, are generally not unique. Nonetheless, Alice's choice of input and output enforces a unique computation among all the possible choices. This choice of g-flow is not communicated to the server and is kept hidden by Alice, who uses it to update the classical instructions sent to Bob. This observation makes it possible for a client to conceal the flow of quantum information from a quantum server classically instructed on what operations to perform. In particular, since a large number of other computations are still compatible with the information Bob receives, the achieved blindness follows from the ambiguity about the flow of information on the graph. Furthermore, our protocol circumvents the scheme-dependent no-go theorem for classical blind quantum computing stated in Ref.~\cite{Morimae2014}. Here, we do not make use of any affine encryption on the client's side, but as mentioned above, we use flow ambiguity to encode a part of the client's computation. As a consequence of this encoding, Protocol~\ref{prot:CDBQCsh} requires multiple rounds of communication between the client and server. This requirement is in direct contrast with the assumptions of Ref.~\cite{Morimae2014}, where only one round of communication is allowed.

\begin{figure}[!b]
\label{9cases}
\begin{center}
\includegraphics[width=\columnwidth]{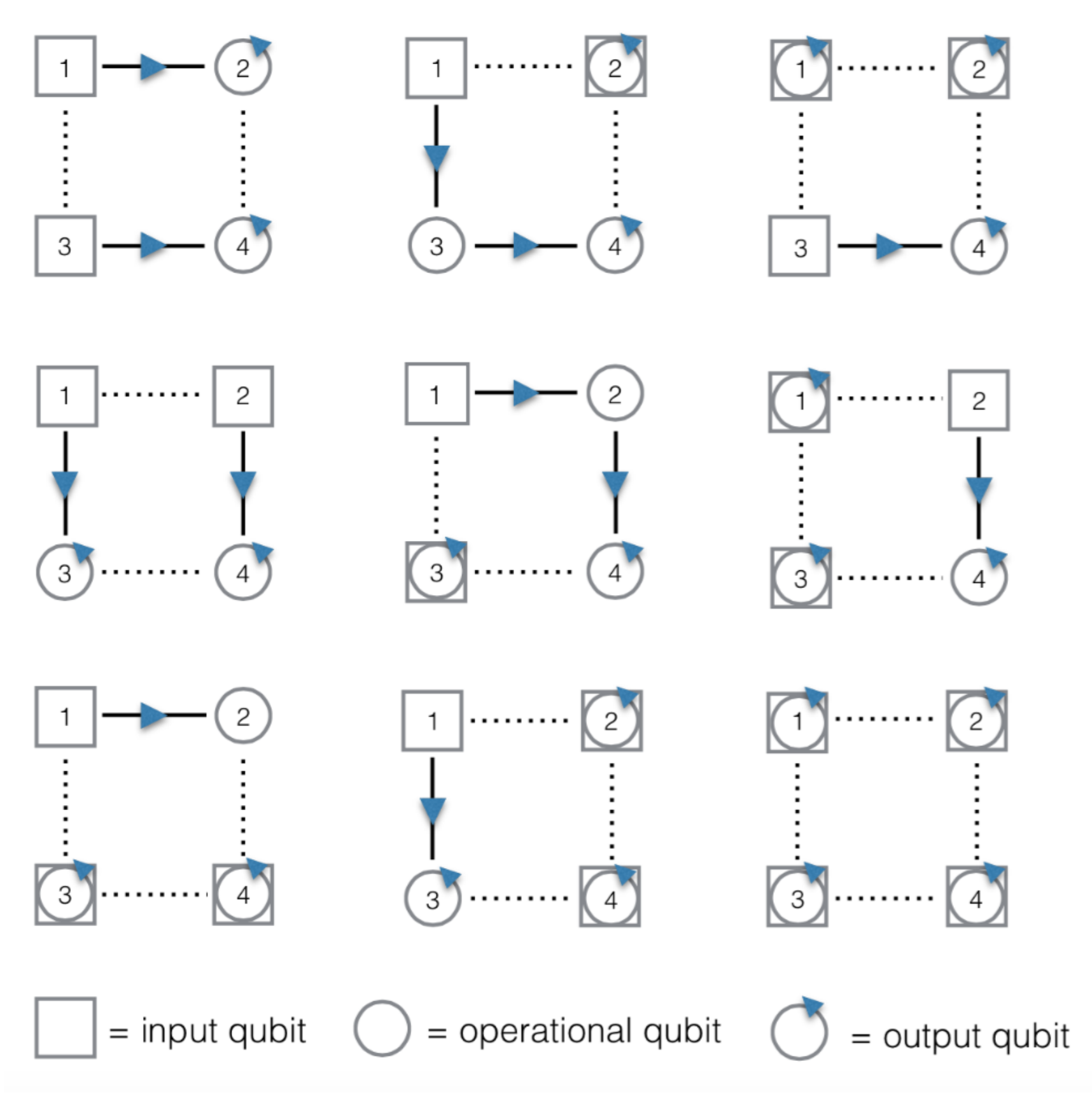}
\end{center}
\caption{List of the nine possible $\mathcal{G}(I,O)_{2,2}$ combinations (and associated patterns) with g-flow for the cluster state $| \text{CS} \rangle_{2,2}$. The arrows indicate the direction of the quantum information flow. Note that overlapping input and output sets are allowed. All the patterns implement unitary embeddings on the input state.\label{fig:2x2flows}}
\end{figure} 

\begin{figure}[t!]
\includegraphics[width=\columnwidth]{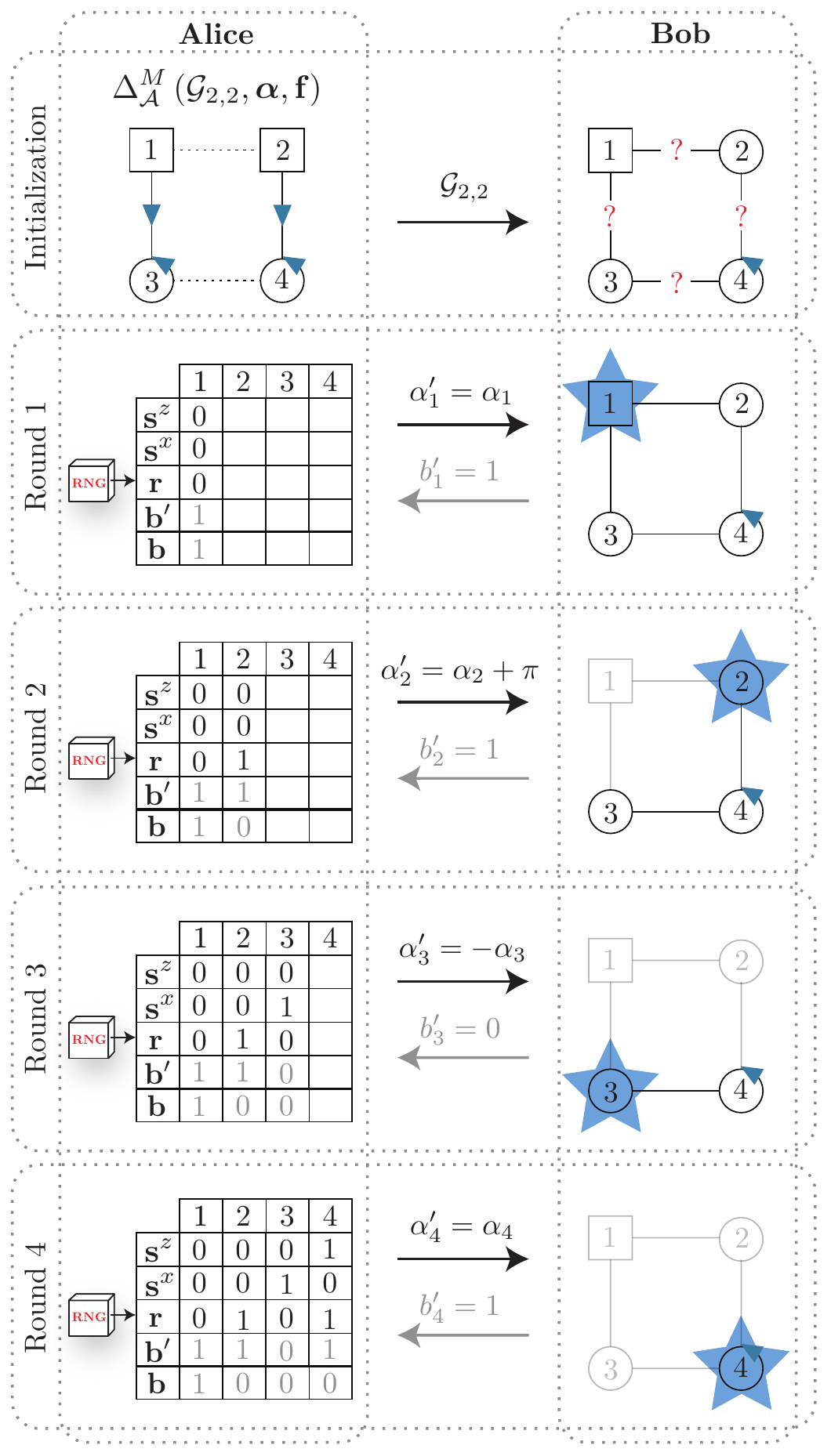}
\caption{Illustration of an exemplary run of Protocol~\ref{prot:CDBQCsh}. At the start of the protocol, Alice's computation is expressed as a measurement pattern on a graph, in this case $\mathcal{G}_{2,2}$. This is communicated to Bob, who prepares the initial state. The computation then proceeds in rounds with Alice computing the relevant entries of $\vec{s}^x$ and $\vec{s}^z$ and, using these together with $r_i$, the measurement angle $\alpha'_i$. The measurement angle $\alpha'_i$ is communicated to Bob, who performs the measurement and returns the result to Alice as $b'_i$. From this, Alice computes $b_i$ as $b'_i \oplus r_i$. This process is repeated until all qubits have been measured. At the end of the protocol, for any fixed transcript of the communication (composed of $\vec\alpha'$ and $\vec{b'}$), it is always possible to find a choice for $\vec{\alpha}$ consistent with the transcript for any choice of $\vec{f}$. In other words, for any of the possible g-flow configurations shown in Fig.~\ref{fig:2x2flows}, Bob can find an~$\vec \alpha$ that would have led to the transcript he recorded, which means any of those g-flows is possible. This ambiguity is responsible for partially hiding Alice's computation from Bob.
\label{fig:gflowex}}
\end{figure}

We can additionally make two important observations. The first is that the circuits implementable on $\mathcal{G}_{n,m}$ are not classically simulable unless $\text{BPP} = \text{BQP}$. This stems from the fact that the cluster state is universal with only XY-plane measurements, as has recently been proven in Ref.~\cite{mantri2017universality}. The above bound provides an exponential lower bound on the number of consistent flows for all cluster states. The second observation is that the computations corresponding to different choices of flow are not equivalent, even when classical post-processing is allowed. This can most easily be seen by considering an example. We consider the simplest case of the $2 \times 2$ plaquette $\ket{\text{CS}}_{2,2}$. In Fig.~\ref{fig:5} we show an example of two different choices of open graphs compatible with the underlying total order of measurements. Both choices satisfy the g-flow conditions and as a result they correspond to two deterministic MBQC   patterns, i.e. two different and well-defined computations. Since in this particular case the input state ($\hat{C}_{Z} \ket{{+}{+}}$) is equivalent, the difference is dictated by the unitary transformation (specified by the measurement angles) that acts on it. As can be seen from Fig.~\ref{fig:5}, the corresponding circuits are different and perform different unitaries. Because of the flow ambiguity and the obfuscation due to the one-time pad, a quantum server that was to perform the measurements following Protocol~\ref{prot:CDBQCsh}, at the end of the procedure, would not have enough information to exactly identify Alice's choice of open graph. For $\mathcal{G}_{2,2}$ there are nine possible flow configurations as expected from Eq.~\eqref{eq:numberofflows}, which are depicted in Fig.~\ref{fig:2x2flows}. Given any fixed transcript of the protocol, for each flow there exists a choice of $\vec \alpha$ such that it is consistent with the transcript. An example run of Protocol~\ref{prot:CDBQCsh} is presented in Fig.~\ref{fig:gflowex} for the $\mathcal{G}_{2,2}$ case. 

As a final comment, it is clear that there exist cases where, for a fixed graph and choice of angles, different choices of flows will correspond to the same computation. For instance, referring to Fig.~\ref{fig:2x2flows}, measuring all qubits with the same angle would give a two-to-one correspondence for some of the computations. However, it is reasonable to conjecture that when the angles are chosen from sets of large cardinality the mapping will be close to one to one. The full characterization of the mapping is left as an open problem.

\section{Discussion and conclusions}
\label{chap6}
Our overall motivation in this work has been to explore the possibility of classically driven blind quantum computation. While this may seem an impossible task, the fact that multiple nonequivalent computations in the MBQC model can yield the same transcript of measurement angles and results, even when the resource state and order of measurements are fixed, allows the tantalizing possibility that it may be possible for a classical user to hide a computation from a quantum server. Protocol~\ref{prot:CDBQCsh} makes use of this flow ambiguity to provide some measure of hiding for quantum computations chosen from certain restricted sets. Our intention in introducing this protocol is not to provide a practical cryptographic protocol but rather to demonstrate that it is indeed possible to hide nonequivalent quantum computations using this flow ambiguity. As such, we concentrate on showing that in a single run of the protocol, the amount of information obtained by the server is bounded, rather than introducing a composable security definition, which is nontrivial given the dependence of the leaked information on the responses of the server. 

Our results provoke a couple of questions. The first and most obvious is whether the flow ambiguity effect can be exploited to hide a universal set of computations even after the measurement angles have been communicated to the server. A second, and perhaps even more important question is whether this phenomenon can be used as a building block for verification of quantum computers by completely classical users.

\acknowledgments

The authors thank Rafael Alexander, Niel de Beaudrap, Michal Hajdu\v{s}ek, Elham Kashefi, Simon Perdrix and Carlos P\'{e}rez-Delgado for interesting discussions and valuable insights. T.F.D.\ thanks Yingkai Ouyang for carefully reading an early version of this manuscript and for his helpful comments. T.F.D.\ also thanks his brother Alessandro Demarie for reminding him why we strive to do science at our best. J.F.F.\ acknowledges support from the Air Force Office of Scientific Research under Grant No.~FA2386-15-1-4082. N.C.M. is supported by the Australian Research Council under Grant No.~DE120102204, by the Australian Research Council Centre of Excellence for Quantum Computation and Communication Technology (Project No.\ CE170100012), and by the U.S.\ Defense Advanced Research Projects Agency (DARPA) Quiness program under Grant No.\ W31P4Q-15-1-0004. This material is based on research funded by the Singapore National Research Foundation under NRF Award NRF-NRFF2013-01.

\bibliographystyle{bibstyleNCM_papers}
\bibliography{allrefs-Tom}

\begin{widetext}
\appendix
\section{Full joint probability for the protocol and conditional entropy bound}
\label{app:jointprob}

\begin{lemma}
\label{lem:condent}
$H(\vec B', \vec A' | \vec A, \vec F) \geq N$ regardless of Bob's strategy.
\end{lemma}
\begin{proof}
We construct the full joint probability for all of the variables in Protocol~\ref{prot:CDBQCsh} and use it to explicitly derive the desired result. Direct dependencies in the joint probability will be limited by causality and the assumptions that Alice's and Bob's laboratories are secure and free of each others' espionage. These limitations are as follows: 
\begin{itemize}
\item The flow bits~$\vec F$ and ideal measurement angles~$\vec A$ directly depend on no other variables. They are inputs to the problem chosen by Alice and can be correlated.
\item Each $\pi$-shift bit~$R_j$ in~$\vec R$ is chosen by flipping a fair coin; thus, it directly depends on no other variables.
\item Alice assigns $A'_j$ based directly on the current~$A_j$, all flow bits~$\vec F$, the current~$R_j$, and any prior decoded bits~$\vec B_{<j}$.
\item Each decoded bit~$B_j$ directly depends only on $R_j$ (the $\pi$-shift bit) and the bit~$B'_j$ received from Bob.
\item Each bit~$B'_j$ that Bob returns to Alice directly depends only on the information Bob has on hand at the time (specifically, $\vec B'_{<j}$ and $\vec A'_{\leq j}$), as well as any (classical or quantum) stochastic strategy he wishes to employ.
\end{itemize}
With these direct-dependency limitations, we can immediately write down the form of the full joint probability for the entire protocol:
\begin{align}
	\Pr(\vec b', \vec \alpha', \vec \alpha, \vec f, \vec b, \vec r)
	&= \Pr(\vec \alpha, \vec f) \prod_{j=1}^N \Pr(b'_j | \vec b'_{<j}, \vec \alpha'_{\leq j}) \Pr(\alpha'_j | \alpha_j, \vec f, \vec b_{<j}, r_j) \Pr(b_j | b'_j, r_j) \Pr(r_j).
\end{align}
Furthermore, we can explicitly write several of these probabilities:
\begin{align}
	\Pr(b_j | b'_j, r_j)
	&= \delta^{b_j}_{b'_j \oplus r_j}
	,
\\
	\Pr(\alpha'_j | \alpha_j, \vec f, \vec b_{<j}, r_j)
	&= \delta^{\alpha'_j}_{G_j(\alpha_j, \vec f, \vec b_{<j}, r_j)}
	,
\\
	\Pr(r_j)
	&= \frac 1 2,
\end{align}
with the deterministic function
\begin{align}
	G_j(\alpha_j, \vec f, \vec b_{<j}, r_j)
	&\coloneqq (-1)^{s^x_j(\vec f, \vec b_{<j})} \alpha_j + \pi s^z_j(\vec f, \vec b_{<j}) + \pi r_j \mod 2\pi
\end{align}
obtained from Eq.~\eqref{eq:alphap}. These hold for all~$j$. At this point, we have the most general form of the full joint probability consistent with the protocol:
\begin{align}
	\Pr(\vec b', \vec \alpha', \vec \alpha, \vec f, \vec b, \vec r)
	&= \frac { \Pr(\vec \alpha, \vec f)} {2^N} \delta^{\vec b}_{\vec b' \oplus \vec r}
		\prod_{j=1}^N \Pr(b'_j | \vec b'_{<j}, \vec \alpha'_{\leq j}) \delta^{\alpha'_j}_{G_j(\alpha_j, \vec f, \vec b_{<j}, r_j)}
	,
\end{align}
where we have left Bob's strategy arbitrary but consistent with the direct-dependency restrictions given above. Marginalizing over~$\vec B$ gives
\begin{align}
	\Pr(\vec b', \vec \alpha', \vec \alpha, \vec f, \vec r)
	&= \sum_{\vec b} \Pr(\vec b', \vec \alpha', \vec \alpha, \vec f, \vec b, \vec r)
\\
	&= \frac { \Pr(\vec \alpha, \vec f)} {2^N}
		\prod_{j=1}^N \Pr(b'_j | \vec b'_{<j}, \vec \alpha'_{\leq j}) \delta^{\alpha'_j}_{G_j(\alpha_j, \vec f, \vec b'_{<j} \oplus \vec r_{<j}, r_j)}
	.
\end{align}
From this joint probability distribution we can compute
\begin{align}
	\Pr(\vec b', \vec \alpha'|  \vec \alpha, \vec f)
	&= \sum_{\vec r}
	\frac 
	{\Pr(\vec b', \vec \alpha', \vec \alpha, \vec f, \vec r)}
	{\Pr(\vec \alpha, \vec f)}
\\
	&= \frac {1} {2^N}
	\sum_{r_1} \dotsm \sum_{r_{N-2}}  \sum_{r_{N-1}} \sum_{r_N} 
	\prod_{j=1}^N \Pr(b'_j | \vec b'_{<j}, \vec \alpha'_{\leq j})
	\delta^{\alpha'_j}_{G_j(\alpha_j, \vec f, \vec b'_{<j} \oplus \vec r_{<j}, r_j)}
\\
	&= \frac {1} {2^N} 
	\left[
	\prod_{j=1}^N \Pr(b'_j | \vec b'_{<j}, \vec \alpha'_{\leq j})
	\right]
	\sum_{r_1} \dotsm \sum_{r_{N-2}} \sum_{r_{N-1}}
	\left[
	\prod_{j=1}^{N-1}
	\delta^{\alpha'_j}_{G_j(\alpha_j, \vec f, \vec b'_{<j} \oplus \vec r_{<j}, r_j)}
	\right]
\nonumber
\\
	&\qquad \qquad \qquad \times
	\underbrace{
	\left(
	\delta^{\alpha'_N}_{G_j(\alpha_N, \vec f, \vec b'_{<N} \oplus \vec r_{<N}, 0)}
	+
	\delta^{\alpha'_N}_{G_j(\alpha_N, \vec f, \vec b'_{<N} \oplus \vec r_{<N}, 1)}
	\right)
	}_{\text{at most one term is nonzero}}
\\
	&\leq \frac {1} {2^N}
	\left[
	\prod_{j=1}^N \Pr(b'_j | \vec b'_{<j}, \vec \alpha'_{\leq j})
	\right]
	\sum_{r_1} \dotsm \sum_{r_{N-2}} \sum_{r_{N-1}}
	\prod_{j=1}^{N-1}
	\delta^{\alpha'_j}_{G_j(\alpha_j, \vec f, \vec b'_{<j} \oplus \vec r_{<j}, r_j)}
\\
	&\leq \frac {1} {2^N}
	\left[
	\prod_{j=1}^N \Pr(b'_j | \vec b'_{<j}, \vec \alpha'_{\leq j})
	\right]
	\sum_{r_1} \dotsm \sum_{r_{N-2}}
	\prod_{j=1}^{N-2}
	\delta^{\alpha'_j}_{G_j(\alpha_j, \vec f, \vec b'_{<j} \oplus \vec r_{<j}, r_j)}
\\
	&\;\;\vdots \nonumber
\\
	&\leq \frac {1} {2^N}
	\left[
	\prod_{j=1}^N \Pr(b'_j | \vec b'_{<j}, \vec \alpha'_{\leq j})
	\right]
	\sum_{r_1}
	\delta^{\alpha'_1}_{G_1(\alpha_1, \vec f, r_1)}
\\
	&\leq \frac {1} {2^N}
	\prod_{j=1}^N \Pr(b'_j | \vec b'_{<j}, \vec \alpha'_{\leq j})
\\
	&\leq \frac {1} {2^N}
\end{align}
In the above, we have repeatedly used the fact that $G_j$ has at most one~$r_j$ that makes it equal to~$\alpha'_j$ for any given $(\alpha_j, \vec f, \vec b_{<j})$. Therefore, substituting the above bound into the conditional entropy formula gives
\begin{align}
	H(\vec B', \vec A' | \vec A, \vec F)
	&= \sum_{\vec \alpha, \vec f} \Pr(\vec \alpha, \vec f)
	H(\vec B', \vec A' | \vec A = \vec \alpha, \vec F = \vec f)
	\geq \sum_{\vec \alpha, \vec f} \Pr(\vec \alpha, \vec f)
	N
	= N,
\end{align}
which was to be proven.
\end{proof}

\end{widetext}

\end{document}